\documentclass[a4paper,UKenglish,cleveref, autoref, thm-restate]{lipics-v2021}
\nolinenumbers
\usepackage{stmaryrd}
\usepackage{tikz-cd}
\usepackage{amsmath, amssymb}
\usepackage{mathtools}
\usepackage{bbold}
\usepackage[scr=euler]{mathalfa}
\usepackage{BOONDOX-calo}

\usepackage{tikz}
   \usetikzlibrary{shapes.geometric, intersections}

\newcommand{\IMP}{\; \Rightarrow \;}

\newcommand{\id}{\mathsf{id}}
\newcommand{\ie}{\textit{i.e.}~}

\newcommand{\vphi}{\varphi}

\newcommand{\MM}{\mathcal{M}}

\newcommand{\es}{\varnothing}

\newcommand{\vv}{\vec{v}}

\newcommand{\Nat}{\mathbb{N}}

\newcommand{\HH}{\mathcal{H}}

\newcommand{\ZZ}{\mathbb{Z}}

\newcommand{\Pauli}{\mathbf{Pauli}}

\newcommand{\Zero}{\mathbf{0}}

\newcommand{\Zen}{\mathscr{Z}}

\newcommand{\rwr}{\underset{R}{\longleftrightarrow}}

\newcommand{\RCG}{R_{\CG}}
\newcommand{\CG}{\mathsf{G}}

\newcommand{\NN}{\mathscr{N}}
\newcommand{\rew}{\rightarrow}
\newcommand{\rCG}{\rew_{\CG}}

\newcommand{\eqr}{\doteq}
\newcommand{\equivR}{\overset{*}{\underset{R}{\longleftrightarrow}}}

\newcommand{\Inv}{\mathscr{I}}
\newcommand{\lcomm}{\llbracket}
\newcommand{\rcomm}{\rrbracket}
\newcommand{\BE}{\mathsf{BE}}
\newcommand{\bd}{\partial}
\newcommand{\mlb}{\{ \! \! |}
\newcommand{\mrb}{| \! \! \}}

\newcommand{\ket}[1]{| #1 \rangle}
\newcommand{\UG}{\mathsf{U}}
\newcommand{\Complex}{\mathbb{C}}
\newcommand{\CM}{\mathsf{C}}
\newcommand{\kv}{\vec{k}}
\newcommand{\lv}{\vec{l}}

\newcommand{\mutr}{\check{\mu}}
\newcommand{\cm}{\mathfrak{so}}
\newcommand{\val}{\mathcal{v}}
\newcommand{\XG}{\mathcal{X}}
\newcommand{\Pnd}{\mathbb{P}_{n,d}}
\newcommand{\dvs}{\, \vdots \,}
\newcommand{\ndvs}{\, \not{\vdots} \,}

\DeclareMathOperator{\Za}{\mathbb{Z}}

\title{Commutation Groups and State-Independent Contextuality} 

\author{Samson Abramsky}{Department of Computer Science, University College London, 66-72 Gower St., London WC1E 6EA, U.K.}{s.abramsky@ucl.ac.uk}{0000-0003-3921-6637}{EPSRC EP/V040944/1 Resources in Computation, UKRI 10050493 FoQaCiA}

\author{\c{S}erban-Ion Cercelescu}{Department of Computer Science, University of Oxford, Parks Road, Oxford OX1 3QD, U.K.}{serban-ion.cercelescu@exeter.ox.ac.uk}{}{}

\author{Carmen-Maria Constantin}{Department of Computer Science, University College London, 66-72 Gower St., London WC1E 6EA, U.K.}{c.constantin@ucl.ac.uk}{0000-0003-4508-9312}{UKRI 10050493 FoQaCiA}

\authorrunning{S. Abramsky, S. Cercelescu and C. Constantin} 

\Copyright{Samson Abramsky, \c{S}erban-Ion Cercelescu, Carmen-Maria Constantin} 



\category{} 

\relatedversion{} 

\EventEditors{}
\EventNoEds{0}
\EventLongTitle{The 9th International Conference on Formal Structures for Computation and Deduction (FSCD 2024)}
\EventShortTitle{FSCD 2024}
\EventAcronym{FSCD}
\EventYear{2024}
\EventDate{July 8--13, 2024}
\EventLocation{Tallinn, Estonia}
\EventLogo{}
\SeriesVolume{}
\ArticleNo{}

\begin{document}

\maketitle

\begin{abstract}
We introduce an algebraic structure for studying state-independent contextuality arguments, a key form of quantum non-classicality exemplified by the well-known Peres-Mermin magic square, and used as a source of quantum advantage.
We introduce \emph{commutation groups} presented by generators and relations, and analyse them in terms of a string rewriting system. There is also a linear algebraic construction, a directed version of the Heisenberg group. We introduce \emph{contextual words} as a general form of contextuality witness. We characterise when contextual words can arise in commutation groups, and explicitly construct non-contextual value assignments in other cases. We give unitary representations of commutation groups as subgroups of generalized Pauli $n$-groups.
\end{abstract}

\section{Introduction}
Contextuality is a key form of non-classicality in quantum mechanics, and is the source of quantum advantage in a range of settings, including measurement-based quantum computation \cite{raussendorf2013contextuality} and shallow circuits \cite{bravyi2018quantum,bravyi2020quantum}. 
In classical physics, observable quantities have well-defined values independently of which measurements are performed.
This is contradicted by the predictions of quantum mechanics \cite{kochen1975problem}, as verified in numerous experiments \cite{bartosik2009experimental,hasegawa2006quantum}.
These say that values can only be assigned \emph{locally}, in measurement contexts, \ie with respect to sets of measurements which can be performed together, providing observational ``windows'' of classical information on the quantum system. These windows may overlap, and will agree on their overlaps (\emph{local consistency}), but it is not possible, on pain of logical contradiction, to glue all these pieces of information together (\emph{global inconsistency}).

The strongest form of this phenomenon is \emph{state-independent contextuality}, where the structure of the observables dictates that contextuality arises for any state. The most famous example of this phenomenon is the Peres-Mermin magic square \cite{mermin1990simple}, which is constructed from the 2-qubit Pauli group\footnote{We recall the definition of the Pauli group in the Appendix.}:
\begin{align*}
\begin{array}{ccccc}
XI & \text{---} & IX & \text{---} & XX \\
| & & | & & | \\
IZ & \text{---}  & ZI & \text{---}  &  ZZ \\
| & & | & & | \\
XZ & \text{---}  &  ZX & \text{---}  & YY
\end{array}
\end{align*}
Here $XI$ denotes the 2-qubit operator $\sigma_x \otimes I$, and similarly for the other entries. One can now calculate that the operators in each row and column pairwise commute, and hence form a valid measurement context. Moreover, the product of each of the rows, and of the first two columns, is $II$; while the product of the third column is $-II$.

We shall now see how to recognize  contextuality in this example. The key point is that this can be done \textit{a priori}, independently of any  observational data.
We ask if there is an assignment of outcomes
$\val : \XG \to  \ZZ_{2}$ , where $\XG$ is the set of operators in the table, subject to the conditions that
\begin{enumerate}
\item if  $p$ and $q$ commute, then $\val(pq) = \val(p) + \val(q)$. 
\item $\val(II) = 0$ and $\val(-II) = 1$.
\end{enumerate}
Such an assignment is called a \emph{non-contextual value assignment}. If no such assignment exists, this yields an example of \emph{contextuality}.
We call this \emph{state-independent}, since it arises purely at the level of the operators in the table, independently of any state.

Note that we only require the homomorphism condition~(1) for \emph{commuting} operators, which correspond to observables that can be performed together, in a common context. This is the key idea introduced by Kochen and Specker in their seminal work on contextuality \cite{kochen1975problem}.

Now assume for contradiction that such an assignment exists.
We obtain the following set of  equations over $\ZZ_2$ from the above table, one for  each row and each column:
\begin{align}
\label{PMeqns}
\begin{array}{ccccccc}
a + b + c & = & 0 & \quad & a + d + g & = & 0  \\
d + e + f & = & 0 & \quad & b + e + h & = & 0 \\
g + h + i & = & 0 & \quad & c + f + i & = & 1 \\
\end{array}
\end{align}
Here $a$ is a variable corresponding to $\val(XI)$, etc.

Since each variable appears twice in the left hand sides, summing over them yields $0$, while summing over the right hand sides yields $1$. This yields the required contradiction.

The justification for assuming the partial homomorphism condition comes from the quantum case, where if $A$ and $B$ are commuting observables and $\psi$ is a common eigenvector of $A$ and $B$, with eigenvalue $v$ for $A$ and $w$ for $B$, then $\psi$ is an eigenvector for $AB$ with eigenvalue $vw$.
Also, $II$ has the unique eigenvalue $+1$, and $-II$ the unique eigenvalue $-1$.\footnote{Note that $\{ +1, -1 \}$ under multiplication is an isomorphic representation of $\ZZ_{2}$, with $0$ corresponding to $+1$ and $1$ to $-1$ under the mapping $i \mapsto (-1)^{i}$.}

We now wish to abstract from the specifics of the Pauli group, and understand the general structure which makes such arguments possible. This leads us to introduce the notion of \emph{commutation group}, to which we now turn.

\section{Commutation groups}
\label{cgsec}
The idea behind commutation groups is that they are built freely from prescribed commutation relations on a set of generators. Commutation relations play a fundamental role in quantum mechanics, the canonical example being the commutation relation between position and momentum (see e.g.~\cite{jordan1986quantum}):
$[p, q]  =  i \hbar \mathbb{1}$.
We can think of a commutation relation as saying that two elements commute \emph{up to a prescribed scalar}. For this to make sense in a group theoretic context, we need an action of a suitable (classical, hence abelian) group of scalars or ``phases'' on the group we are constructing. We are interested here in finite group constructions, so we shall work over the finite cyclic groups $\ZZ_d$, $d \geq 2$.

Given a finite set $\XG$ of generators, we define a \emph{commutator matrix} to be a map $\mu : \XG^{2} \to \ZZ_d$ which is skew-symmetric, meaning that $\mu(x,y) = -\mu(y,x)$ for all $x, y \in X$. We also assume that $\mu(x,x) = 0$ for all $x \in \XG$.\footnote{Note that if $d$ is even, this does not follow automatically from skew-symmetry.}

We shall describe the construction of commutation groups from  commutator matrices in two ways: by generators and relations, and by a linear algebraic construction. Both are useful, and convey different intuitions.

\subsection{Commutation groups by generators and relations}
We briefly review the standard notion of \emph{presentation of a monoid by generators and relations} $\langle \XG \, | \, R \rangle$.
We form the free monoid $\XG^*$, and quotient it  by the congruence induced from the relations $R \subseteq \XG^* \times \XG^*$. Explicitly,
we define a symmetric relation ${\rwr} \, \subseteq \, \XG^* \times \XG^*$ by $s \rwr t$ iff there is $(u,v) \in R \cup R^{-1}$ such that, for some $w_1,w_2 \in \XG^*$,  $s = w_1 u w_2$, and $t = w_1 v w_2$.  We then take the reflexive transitive closure $\equivR$. This is a monoid congruence, and the quotient $M = \XG^*/{\equivR}$ is the presented monoid.

\textbf{Notation} We write relations as $u \eqr v$. We write the empty sequence, which forms the identity element of the free monoid, as $1$.

Given a commutator matrix $\mu : \XG^{2} \to \ZZ_d$, we define a set of relations $\RCG$ over the generators $\XG \sqcup \ZZ_d$ (using $\sqcup$ for disjoint union),  where we write $J_k$ for the generator corresponding to $k \in \ZZ_d$, and:
\begin{itemize}
\item We have relations $R_{\mu} := \{ xy \eqr J_{\mu(x,y)} yx \mid x, y \in \XG \}$.
\item We have $R_{J} := \{ J_0 \eqr 1 \} \cup \{ J_{k}J_{k'} \eqr J_{k+k'} \mid k, k' \in \ZZ_d \} \cup \{ J_k x \eqr x J_k \mid x \in \XG, k \in \ZZ_d \}$.
\item We have $R_{d} := \{ x^d \eqr 1 \mid x \in \XG\}$.
\item Finally, $\RCG := R_{\mu} \cup R_{J} \cup R_{d}$.
\end{itemize}
The resulting monoid $\CG(\mu) := \langle \XG \sqcup \ZZ_d \mid \RCG \rangle$ is in fact a group, since every generator has an inverse. We call it the \emph{commutation group} generated by $\mu$.

The $J$-relations ensure that there is an isomorphic copy of $\ZZ_d$ in the centre of the group.
The key relations are the commutation relations $xy \eqr J_{\mu(x,y)} yx$. Note that these are \emph{directional}, since by skew-symmetry of $\mu$, if $\mu(x,y) = k$, then $yx \eqr J_{-k}xy$. Thus moving $x$ right past $y$ has the opposite ``cost'' to moving $x$ left past $y$.
This suggests that we can analyze $\CG(\mu)$ by a directed \emph{string rewriting system}.

To do this, we fix a linear ordering $x_1 < \cdots < x_n$ on $\XG$.\footnote{As we shall see, the choice of ordering is immaterial, leading to isomorphic results.}
Relative to this ordering, elements of $\CG(\mu)$ can be represented as ordered multisets over $\XG$ with multiplicities strictly less than $d$, together with a ``global phase'' from $\ZZ_d$.
Explicitly, we define $\NN$ to be the set of all expressions $J_{k}x_{1}^{k_1} \cdots x_{n}^{k_n}$, with $k \in  \ZZ_d$, and $0 \leq k_i < d$, $1 \leq i \leq n$.
There is an evident bijection $\NN \cong \ZZ_d \times \ZZ_{d}^{n}$. Thus $\NN$ has cardinality $d^{n+1}$.

We  now define a string rewriting system on $\XG \sqcup \ZZ_d$, obtained by orienting a subset of the relations $\RCG$, determined by the chosen linear order on $\XG$:
\begin{itemize}
\item $\rew_{\mu} \, := \, \{ xy \rew J_{\mu(x,y)} yx \mid x > y \}$.
\item $\rew_{J} \, := \, \{ J_{0} \rew 1 \} \cup \{ J_{k}J_{k'} \rew J_{k + k'} \mid k, k' \in \ZZ_d \} \cup \{ xJ_{k} \rew J_{k}x \mid x \in \XG, k \in \ZZ_d \}$.
\item $\rew_{d} \, := \, \{ x^d \rew 1 \}$.
\item $\rCG \, := \, {\rew_{\mu}} \cup {\rew_{d}} \cup {\rew_{J}}$.
\end{itemize}
This induces a relation on $(\XG \sqcup \ZZ_d)^*$ by $s \rew t$ iff for some $u \rCG v$, for some $w_1, w_2 \in (\XG \sqcup \ZZ_d)^*$, $s = w_1 u w_2$, and $t = w_1 v w_2$.

\begin{theorem}
\label{contermthm}
The rewrite system $\rCG$ is confluent and normalizing. The set of normal forms is $\NN$ (up to identification of $J_0$ and $1$).
\end{theorem}
\begin{proof}
Given a word $s \in (\XG \sqcup \ZZ_d)^*$, we define:
\begin{itemize}
\item An $\XG$-inversion in $s$ is $(u,v,w,x,y)$ such that $s = uxvyw$, and $x > y$.
\item A $J$-inversion in $s$ is $(u,v,w,x,k)$ such that $s = uxvJ_{k}w$.
\end{itemize}
We define a function $\vphi : (\XG \sqcup A)^* \to \Nat \times \Nat \times \Nat$ by $\vphi(s) = (n,m,l)$, where $n$ is the number of $\XG$-inversions in $s$, $m$ is the number of $J$-inversions, and $l$ is the length of $s$. 

We now observe that for each rewrite  $s \rew t$ in the above system $\rCG$, $\vphi(s) \succ \vphi(t)$ in the lexicographic ordering on $\Nat \times \Nat \times \Nat$. Indeed, the $\mu$ relations decrease the number of $\XG$-inversions, the $J$-commutation rule decreases the number of $J$-inversions while not increasing the number of $\XG$-inversions, and the remaining rules decrease length while not increasing the number of inversions. Since this ordering is well-founded, it follows that $\rCG$ is normalizing.

By Newman's Lemma, it now suffices to show that $\rCG$ is weakly confluent. This is verified straightforwardly by examining the critical pairs.

Firstly, consider $x > y > z$, $\mu(x,y) = a$, $\mu(y,z) = b$, $\mu(x,z) = c$:
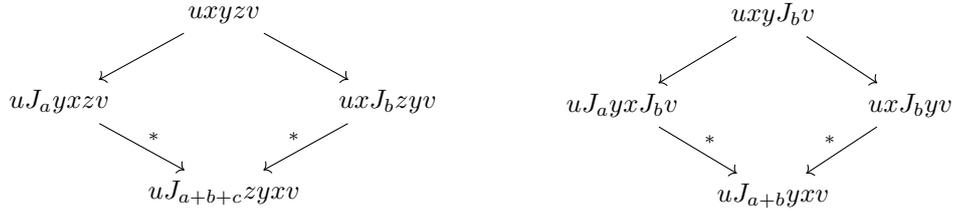
\[ \begin{tikzcd}[column sep=tiny]
& uxyzv \ar[dr] & \\
uJ_{a}yxzv \ar[ur, leftarrow] \ar[dr, "*"] & & uxJ_{b}zyv \ar[dl, "*"'] \\
& uJ_{a+b+c}zyxv & \\
\end{tikzcd}
\qquad \qquad
\begin{tikzcd}[column sep=tiny]
& uxyJ_{b}v \ar[dr] & \\
uJ_{a}yxJ_{b}v \ar[ur, leftarrow] \ar[dr, "*"] & & uxJ_{b}yv \ar[dl, "*"'] \\
& uJ_{a+b}yxv & \\
\end{tikzcd}
\]
Next, two cases involving $J$-generators:
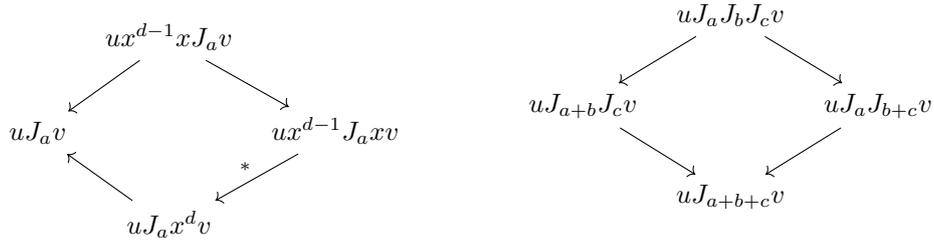
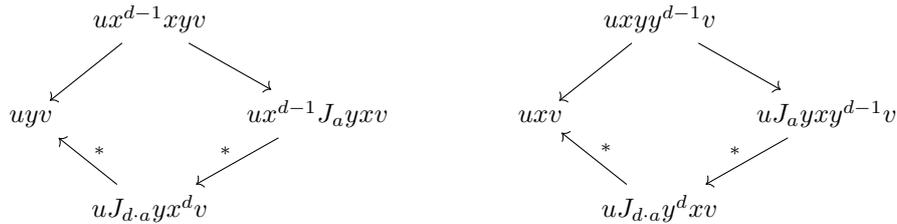
\[ \begin{tikzcd}[column sep=tiny]
& ux^{d-1}xJ_{a}v \ar[dl] \ar[dr] &  \\
uJ_{a}v \ar[dr, leftarrow] &  & ux^{d-1}J_{a}xv  \ar[dl, "*"'] \\
& uJ_{a}x^{d}v   &
\end{tikzcd}
\qquad \qquad
\begin{tikzcd}[column sep=tiny]
& uJ_{a}J_{b}J_{c}v \ar[dr] & \\
uJ_{a+b}J_{c}v \ar[ur, leftarrow] \ar[dr] & & uJ_{a}J_{b+c}v \ar[dl] \\
& uJ_{a+b+c}v & \\
\end{tikzcd}
\]
Finally:
\[ \begin{tikzcd}[column sep=tiny]
& ux^{d-1}xyv \ar[dl] \ar[dr] &  \\
uyv \ar[dr, leftarrow, "*"] &  & ux^{d-1}J_{a}yxv  \ar[dl, "*"'] \\
& uJ_{d \cdot a}yx^{d}v   &
\end{tikzcd}
\qquad \qquad
\begin{tikzcd}[column sep=tiny]
& uxyy^{d-1}v \ar[dl] \ar[dr] &  \\
uxv \ar[dr, leftarrow, "*"] &  & uJ_{a}yxy^{d-1}v  \ar[dl, "*"'] \\
& uJ_{d \cdot a}y^{d}xv   &
\end{tikzcd}
\]
Note that $d \cdot a = 0 \mod d$, justifying the final legs.

\end{proof}

By virtue of this theorem, we can define a function $\theta : (\XG \sqcup \ZZ_d)^* \to \NN$, which returns the normal form of a word. 
Note that, if $w \rew^{*} w'$, then by confluence, $\theta(w) = \theta(w')$. 

We can use this function to define an equivalence on $(\XG \sqcup \ZZ_d)^*$ by $s \simeq t$ iff $\theta(s) = \theta(t)$. This equivalence is in fact a congruence, since if $\theta(u) = s = \theta(u')$ and $\theta(v) = t = \theta(v')$, then, using confluence, $\theta(uv) = \theta(st) = \theta(u'v')$.
\begin{proposition}
\label{equivprop}
For all $s, t \in (\XG \sqcup \ZZ_d)^*$, $s \simeq t$ iff  $s \eqr t$.
\end{proposition}
\begin{proof}
The left-to-right implication follows immediately since $\rCG \, \subset \, \RCG$.
For the converse, it suffices to show that $s \simeq t$ for all relations $s \eqr t$ in $\RCG$, since $\eqr$ is the least congruence containing these relations.

Consider firstly $xy \eqr J_{k}yx$, where $k = \mu(x,y)$. There are two cases:
\begin{enumerate}
\item If $x < y$, then $xy \in \NN$, and $J_{k}yx \rew J_{k}J_{-k}xy \rew^* xy$.
\item If $x > y$, then $J_{k}yx \in \NN$, and $xy \rew J_{k}yx$.
\end{enumerate}
The other relations are verified similarly.
\end{proof}
We now define a monoid with carrier $\NN$. Note that $1$ and $J_k$, $k \in \ZZ_d$, are in $\NN$. We define the multiplication by $u \cdot v := \theta(uv)$.
\begin{proposition}
$(\NN, {\cdot}, 1)$ is a monoid.
\end{proposition}
\begin{proof}
We need to verify associativity. This follows from 
\begin{equation}
\label{thetaeq}
\theta(\theta(uv)w) = \theta(uvw) = \theta(u\theta(vw))
\end{equation}
which in turn follows from confluence.
\end{proof}
We now define a map $h : \CG(\mu) \to \NN$ by $h([w]) = \theta(w)$.
\begin{theorem}
\label{srsthm}
The map $h$ is well-defined, and is a monoid isomorphism $h : \CG(\mu) \cong \NN$.
\end{theorem}
\begin{proof}
If $u \eqr v$, then by Proposition~\ref{equivprop}, $\theta(u) = \theta(v)$. Thus $h$ is well-defined. The fact that it preserves multiplication follows from $\theta(uv) = \theta(\theta(u)\theta(v))$, which follows from confluence. If $w \in \NN$, then $h([w]) = \theta(w) = w$. Thus $h$ is surjective. Finally, if $h([u]) = h([v])$, then $u \simeq v$, so by Proposition~\ref{equivprop}, $[u] = [v]$.
\end{proof}

We now come to a key property for applications to contextuality.
\begin{theorem}
The internal $\ZZ_d$-action given by the $J$-generators is faithful: if $J_k \eqr J_{k'}$ in $\CG(\mu)$, then $k = k'$.
\end{theorem}
\begin{proof}
This is immediate from the isomorphic representation given by $\NN$, since if $k \neq k'$, $J_k$ and $J_{k'}$ are distinct normal forms.
\end{proof}

The parameter $d$ plays a double role in the commutation groups, defining the order of the generators by the relations $x^d = 1$, and also the abelian ``phase group'' $\ZZ_d$ acting on the commutation group. We used this double role of $d$ in proving confluence for the rewriting system. This assumption is in fact necessary to obtain a confluent system with a faithful action, as the following example shows.
\begin{example}
We assume the relations $x^d \eqr 1$ for the generators. Consider the word $w \equiv yxy^{d-1}x^{d-1}$, and let $a = \mu(x,y)$. Then $w \eqr J_{(d-1) \cdot a} y^d x^d \eqr J_{(d-1) \cdot a}$, and also $w \eqr J_{-a} x y^{d} x^{d-1} \eqr J_{-a}  x^{d} \eqr J_{-a}$. Thus to maintain confluence and faithfulness of the action, we require $(d-1) \cdot a = -a$, and hence $d \cdot a = 0$.
\end{example}

\subsubsection{Comparison with solution groups}

In \cite{cleve2017perfect} another way of abstracting from the Peres-Mermin square and similar constructions is pursued, leading to the introduction of \emph{solution groups}.
These groups are specified by sets of equations of similar form to~(\ref{PMeqns}). The generators appearing together in an equation, and only these, are specified to commute.
These groups are shown in \cite{cleve2017perfect} to control the question of whether there is a quantum realization for these equations. Importantly, this is shown to be equivalent to the existence of quantum perfect strategies for Alice-Bob non-local games.

A remarkable result of Slofstra \cite{slofstra2020tsirelson} shows that solution groups, even over $\ZZ_2$, are extremely expressive (\ie ``wild''). Every finitely presentable group can be embedded in a solution group. It follows immediately that the word problem for solution groups, and hence the quantum realization questions, are undecidable.

By contrast, commutation groups are highly tractable. By Theorem~\ref{srsthm}, they are always finite. From the proof of termination in Theorem~\ref{contermthm}, we see that reduction to normal form, and hence the decision procedure for the word problem, is at most quadratic in the length of the word. As we shall see later, every commutation group admits a faithful unitary representation.

\subsection{Linear algebraic construction of commutation groups}
The characterization of commutation groups in Theorem~\ref{srsthm} suggests another description. We shall now use the fact that $\ZZ_d$ is not just an abelian group, but a commutative ring with unit. We can write a commutator matrix, with a chosen order on the set of generators, as an $n \times n$ matrix with entries in $\ZZ_d$. 
We write $\cm(n,\Za_d)$ for the set of all $n \times n$ commutator matrices (skew-symmetric and zero on the diagonal) over $\Za_d$. Given a commutator matrix $\mu$, we write $\mutr$ for its lower triangular part, so that $\mu = \mutr - \mutr^{\intercal}$.

An  $n \times n$ matrix $M$ over $\ZZ_d$ defines a bilinear form on the free $\ZZ_d$-module $\ZZ_d^n$, by $M(\kv,\lv) := \kv^{\top} M \lv$. Now given $\mu \in \cm(\ZZ_d, n)$, we define a group $H(\mu)$ with carrier $\ZZ_d \times \ZZ_d^n$. The group product is defined by 
\[ (k, \kv) \cdot (l, \lv) \; = \; (k + l + \mutr(\kv, \lv), \, \kv + \lv) . \]
Thus it is precisely the phase factor $\mutr(\kv, \lv)$ which makes the group non-commutative.

The associativity of the product follows from bilinearity. The unit is $(0,0)$. The inverse of $(k, \kv)$ is $(-k - \mutr(\kv, -\kv), -\kv)$.
\begin{proposition}
\label{GHprop}
For any $\mu \in \cm(n, \ZZ_d)$, $H(\mu) \cong \CG(\mu)$.
\end{proposition}
\begin{proof}
By Theorem~\ref{srsthm}, the carriers are in evident bijection: $J_{k}x_{1}^{k_1} \cdots x_{n}^{k_n} \, \leftrightarrow \,(k, (k_1, \ldots , k_n))$. We just have to check that the group product is preserved. The only non-imediate part of this is to check that the phase factors agree.

Suppose in $\NN$ we have normal forms with vector parts $u = x_1^{k_1}\cdots x_n^{k_n}$ and $v = x_1^{l_1}\cdots x_n^{l_n}$. To combine them into $\theta(uv)$, with vector part $x_1^{k_1 + l_1}\cdots x_n^{k_n + l_n}$, we must move $l_1$ copies of $x_1$ over $k_n$ copies of $x_n$, each with a cost of $\mu(x_n,x_1)$; and similarly for the occurrences of $x_{n-1} , \ldots , x_2$ in $u$, with total cost
$\sum_{i > 1} k_i \mu(x_i, x_1)  l_1$. A similar analysis applies to the occurrences of $x_2, \ldots , x_{n-1}$ in $v$, leading to a total cost of $\sum_{i > j} k_i \mu(x_i, x_j)  l_j$. This is exactly $\kv^{\top} \mutr \lv = \mutr(\kv,\lv)$.
\end{proof}
Note that we would get the same result if we moved the vector part of $u$ rightwards over the vector part of $v$. The choice of left/right orientation is just a convention. On the other hand, the use of $\mutr$ rather than $\mu$ is significant. As we will see in the next section, using $\mu$ would render the structure useless for our purpose of analyzing contextuality.

However, we do retrieve $\mu$ as the group-theoretic commutator in $H(\mu)$.
\begin{proposition}
\label{commutatorprop}
Given $g = (k, \kv) ,h = (l, \lv) \in H(\mu)$, their group theoretic commutator is given by $[g,h] := ghg^{-1}h^{-1} = (\mutr(\kv,\lv) - \mutr(\lv,\kv),0) = (\mu(\kv,\lv),0)$.
In terms of $\CG(\mu)$, for $u, v \in X^*$, $[u,v] \eqr J_{k-l}$, where $\theta(uv) \eqr J_k w$, $\theta(vu) \eqr J_l w$.
\end{proposition}

As Proposition~\ref{GHprop} makes clear, commutation groups are very close to the (discrete version of) the Heisenberg or Heisenberg-Weyl groups \cite{semmes2003introduction}, and their close relatives the Pauli groups.
The novelty lies mainly in our combinatorial mode of presentation of commutation groups, which we will make use of in our analysis of contextuality arguments.
It should be noted, though, that the direct equivalent of the usual Heisenberg group construction in our setting would be to use the full commutator matrix $\mu$.\footnote{There is a notion of polarized Heisenberg group, but this is isomorphic to the usual presentation.} 
As we have mentioned, using $\mu$ would yield a non-isomorphic construction, which would not be useful for analyzing contextuality.
This perhaps suggests that we can  think of commutation groups as a  \emph{directed} version of Heisenberg groups.

\section{Contextuality arguments in commutation groups}
The commutation group has an evident short exact sequence
\begin{equation}
\label{seseq}
\begin{tikzcd}
\Zero \arrow[r] & \ZZ_d  \arrow[r, "i"] & H(\mu)  \arrow[r, "\pi_2"] & \ZZ_d^n  \arrow[r] & \Zero
\end{tikzcd}
\end{equation}
where $i(k) = (k,0)$. This says that it is a non-abelian group extension of $\ZZ_d^n$ by $\ZZ_d$. The image of $\ZZ_d$ lies in the centre of $H(\mu)$, so the extension is central.
Because of the non-commutativity of $H(\mu)$, it is easy to see that there is no left-splitting of this extension, \ie a homomorphism $l : H(\mu) \to \ZZ_d$ such that $l \circ i = \id_{\ZZ_d}$. One could say that this simple observation is essentially a form of von Neumann's much criticised No-Go theorem for hidden variables \cite{mermin2018homer}. The point of the criticism is that it is not reasonable to ask for a splitting which preserves non-commuting products.

Following Kochen-Specker \cite{kochen1975problem} and the huge literature on ensuing developments, we want to consider only assignments to \emph{observationally accessible contexts}, \ie those constructed from commuting products. A general setting for capturing this idea is provided by \emph{compatible monoids}, introduced in \cite{aasnaess2022comparing} with different terminology. 
A compatible monoid is a structure $(M, \odot, {\cdot}, 1)$, where $\odot$ is a reflexive, symmetric relation on $M$, of ``compatibility'' or ``comeasurability'', and $\cdot : \odot \to M$ is a partial binary operation with domain $\odot \subseteq M^2$, such that: 
\begin{itemize}
\item $x \odot y \IMP x\cdot y = y \cdot x$, 
\item $x \odot 1$ for all $x \in M$, and $x \cdot 1 = x$, 
\item if $x \odot y$, $x \odot z$, and $y \odot z$, then $x \odot (y \cdot z)$ and $(x \cdot y) \odot z$, and $(x \cdot y) \cdot z = x \cdot (y \cdot z)$.
\end{itemize}
Homomorphisms of compatible monoids are maps which preserve the compatibility relation, and the monoid operations when defined.

Any monoid $M$ defines a compatible monoid with the same carrier, with $x \odot y$ iff $xy = yx$ in $M$. We will be interested in the compatible submonoid of $M$ generated by a set $S \subseteq M$. This is the least set $T$ containing $S \cup \{ 1 \}$, and such that, whenever $u, v \in T$ and $u \odot v$, then $u \cdot v \in T$.
In particular, we will apply this to $\CG(\mu)$ with respect to the generators $\XG \sqcup \ZZ_d$. We will write $\CM(\mu)$ for this compatible submonoid of $\CG(\mu)$.

In terms of  $H(\mu)$, we can identify the generators as follows: $k \in \ZZ_d$ can be identified with the scalar $(k, 0)$, while the generators $\XG$ can be identified with the standard basis $E$ of the free module $\ZZ_d^n$, $x_i \, \leftrightarrow \, e_i := [\underbrace{0, \ldots , 0}_{i-1}, 1, \underbrace{0, \ldots , 0}_{n-i}]^{\top}$.

We obtain a short exact sequence for $\CM(\mu)$:
\begin{equation}
\label{seseqCM}
\begin{tikzcd}
\Zero \arrow[r] & \ZZ_d  \arrow[r, "i"] & \CM(\mu)  \arrow[r, "p"] & P  \arrow[r] & \Zero
\end{tikzcd}
\end{equation}
Here $p$ is the restriction of the second projection to $\CM(\mu)$, and $P$ is its image.

A non-contextual value assignment for the commutation group $\CG(\mu)$ is exactly a left splitting of this short exact sequence: \ie a homomorphism $l : \CM(\mu) \to \ZZ_d$ such that $l \circ i = \id_{\ZZ_d}$.
If no such left splitting exists, then we say that $\CG(\mu)$ exhibits \emph{state-independent contextuality}.

\subsection{Contextual words}
We can distill the essential features of ``parity proofs'' such as the one given for the Peres-Mermin square into a notion of \emph{contextual word}, which provides a witness for state-independent contextuality. This notion was introduced, somewhat informally, in the concrete context of Pauli groups over qubits in \cite{kirby2021variational}, but can be formulated generally for any commutation group $\CG(\mu)$. A contextual word for $\CG(\mu)$ is given by a triple $(w, \beta, k)$ such that:
\begin{itemize}
\item $w \in \XG^+$.
\item The number of occurrences of each generator $x \in \XG$ in $w$ is a multiple of $d$.
\item $\beta$ is a bracketing of $w$, witnessing that it is in $\CM(\mu)$.
\item $w \eqr J_k$, where $k \neq 0$.
\end{itemize}
Bracketings are defined inductively by
\[ \beta \in \BE \;\; ::= \;\; x \mid (\beta_1, \beta_2) . \]
We define $\bd : \BE \to \XG^+$ by $\bd(x) = x$, $\bd(\beta_1, \beta_2) = \bd(\beta_1)\bd(\beta_2)$. If $\bd(\beta) = w$, then $w$ is the word bracketed by $\beta$.
A bracketing $\beta$ provides a witness for $w = \bd(\beta) \in \CM(\mu)$ if, for every $(\beta'_1, \beta'_2)$ occurring in $\beta$, with $u = \bd(\beta'_1)$ and $v = \bd(\beta'_2)$, $uv \eqr vu$ in $\CG(\mu)$.


\begin{theorem}
\label{cwsicth}
A contextual word exists for $\CG(\mu)$ if and only if it is state-independently contextual.
\end{theorem}

\begin{proof}
If $(w, \beta, k)$ is a contextual word over $S$, assume for a contradiction that $l : \CM(\mu) \to \ZZ_d$ is a non-contextual value assignment, \ie a left splitting of~(\ref{seseqCM}). The bracketing $\beta$ witnesses that $w \in \CM(\mu)$. By the homomorphism property, $l(w) = \sum_i l(x_i)$, where $w = x_1 \cdots x_n$. Since each $x \in \XG$ occurs with  multiplicity $kd$ in $w$ for some $k \geq 0$, $l(w) = 0 \pmod d$. However, we also have $w \eqr J_k$, so we must have $l(w) = l(J_k) = k \neq 0$, yielding the required contradiction.

Conversely, if no contextual word exists over $S$ then we can construct a non-contextual value assignment $l$. We start by restricting our attention to a compatible submonoid of $\CM(\mu)$, which we shall denote by $\CM'(\mu)$, generated by taking the commutative closure of all the non-scalar generators of $H(\mu)$, and proving the important claim that if $g = (k_g, \vv), h = (k_h, \vv) \in C'(\mu)$ have the same vector component then they must also have the same scalar component $s(\vv):=k_g=k_h$. Note firstly that $g$ commutes with $h$: by Proposition~\ref{commutatorprop}, their commutator is $[g,h] = (\mutr(\vv,\vv) - \mutr(\vv,\vv), \vec{0}) = (0, \vec{0})$.

Next, observe that the inverse of $h$ must also belong to $\CM'(\mu)$: all the products $h^n$ will belong to $\CM'(\mu)$, since $h$ commutes with itself and, since $\CM'(\mu)$ is a compatible submonoid of a finite group, one of these powers, namely the order of $h$, must equal the identity in the original group, which is also the identity of the compatible submonoid. Let $h^{-1} = h^{o(h) - 1} = (l, \vec{w})$.
Then $h\cdot h^{-1}=(0,\vec{0})$, so $\vec{v} + \vec{w} = \vec{0}$, and $k_h + l + \mutr(\vec{v}, \vec{w}) = 0$.
 Thus $g\cdot h^{-1}= (k_g + l + \mutr(\vec{v}, \vec{w}), \vec{0}) = (k_g-k_h,\vec{0})$. Note that $g$ commutes with $h^{-1}$, since it commutes with $h$.

Now if $k:=k_g-k_h$ were non-zero then $J_k \neq 0$, which implies that the product of generators giving rise to $g\cdot h^{-1}$ forms a contextual word, since the compatible monoid 
$\CM'(\mu)$ is formed by taking the closure of the set of generators under commuting products. Therefore, we conclude that $k_g$ must equal $k_h$, so the vector component $\vv$ of any element $g\in\CM'(\mu)$ uniquely determines that element's scalar component $s(\vv)$. Moreover, 
given any commuting elements $(s(\vec{v}), \vec{v}), (s(\vec{w}), \vec{w}) \in \CM'(\mu)$, we have $(s(\vec{v}), \vec{v}) \cdot (s(\vec{w}), \vec{w}) = (s(\vec{v}) + s(\vec{w}) + \mutr(\vec{v}, \vec{w}), \vec{v}+\vec{w}) =
 (s(\vec{v}+\vec{w}), \vec{v}+\vec{w})$.

This canonical scalar value assignment allows us to define a left splitting of the compatible monoid $\CM(\mu)$, which is generated by $\CM'(\mu)\cup \{(k,\vec{0}),k\in\ZZ_d\}$. Let 
$l:\CM(\mu)\rightarrow \ZZ_d$ act on a general element $(k,\vv)$ by $l(k,\vv):=k-s(\vv)$. This is clearly well-defined and moreover it is also a non-contextual assignment. 

To see this is indeed the case note that $l$ is a compatible monoid homomorphism: since the identity belongs to $\CM'(\mu)$, we must have $s(\vec{0})=0$. Hence $l(k,\vec{0})=k$ for all $(k,\vec{0})\in\CM(\mu)$. Also,
for commuting elements $(k_1, \vv_1), (k_2, \vv_2) \in \CM(\mu)$:
\begin{align*}
    l((k_1,\vv_1)\cdot (k_2,\vv_2))&=l(k_1+k_2 + \mutr(\vv_1, \vv_2), \vv_1+\vv_2)\\
    &=k_1+k_2 + \mutr(\vv_1, \vv_2)-s(\vv_1+\vv_2)\\
    &=k_1+k_2 + \mutr(\vv_1, \vv_2) -(s(\vv_1)+s(\vv_2) + \mutr(\vv_1,\vv_2))\\
    &=k_1-s(\vv_1)+k_2-s(\vv_2)\\
    &=l(k_1,\vv_1)+l(k_2,\vv_2)
\end{align*}
and so $l$ is a left splitting.
\end{proof}

\begin{example}
\label{KLex}
Consider the following commutator matrices over $\ZZ_2$
\[
\mu_1 \;\; = \;\; \begin{bmatrix}
0 & 0 & 0 & 1 \\
0 & 0 & 1 & 0 \\
0 & 1 & 0 & 0 \\
1 & 0 & 0 & 0
\end{bmatrix}
\qquad
\mu_2 \;\; = \;\; \begin{bmatrix}
0 & 0 & 0 & 1 \\
0 & 0 & 1 & 1 \\
0 & 1 & 0 & 0\\
1 & 1 & 0 & 0
\end{bmatrix}
\qquad
\mu_3 \;\; = \;\; \begin{bmatrix}
0 & 0 & 0 & 1 \\
0 & 0 & 1 & 1 \\
0 & 1 & 0 & 1 \\
1 & 1 & 1 & 0
\end{bmatrix}
\]
with generators $a < b < c < d$. Then  $(w_i, \beta_i, 1)$ is a contextual word for $\CG(\mu_i)$, with 
\[ \begin{array}{ccccccl}
w_1 & = & abdccabd   & \quad & \beta_1 & = &   ((ab)(dc))((ca)(bd)) \\
w_2 & = & bdccaabd   & \quad & \beta_2 & = & (b(dc))((ca)((ab)d))  \\
w_3 & = & dcabbadc  & \quad & \beta_3 & = &  (d(ca))(b(((ba)d)c)) \\
\end{array}
\]
\end{example}

\begin{example}
We show that the Peres-Mermin square arises in commutation groups.

Firstly, we show that the tensor product construction underlying the extension of the Pauli group to $\Pauli_n$ has a simple form in terms of group presentations.
Given a commutator matrix $\mu : \XG^{2} \to \ZZ_d$, we define $\mu_2 : (\XG + \XG)^2 \to \ZZ_d$ by $\mu_2 (x_{i}, y_{i}) = \mu(x,y)$, $\mu_{2}(x_{i}, y_{j}) = 0$, $i \neq j$.
Thus elements in different copies of $\XG$ commute with each other.

We now consider the commutator matrix with $\mu(x,y) = 1$ for generators $x,y$. 
We can define the following Peres-Mermin square over $\CG(\mu_2)$:
\begin{align*}
\begin{array}{ccccc}
x_{1} & \text{---} & x_{2} & \text{---} & x_{1}x_{2} \\
| & & | & & | \\
y_{2} & \text{---}  & y_{1} & \text{---}  &  y_{1}y_{2} \\
| & & | & & | \\
x_{1}y_{2} & \text{---}  &  y_{1}x_{2} & \text{---}  & J_{1}(x_{1}y_{1})(x_{2}y_{2})
\end{array}
\end{align*}
We can verify that exactly the same algebraic properties hold for this square as in the concrete example: each row and column pairwise commutes, the product of each row and the first two columns is $1$, the product of the third column is $-1$ (or more pedantically, $J_{1}$ in additive notation).

We can extract a contextual word from this construction: $((x_{1}y_{2})(y_{1}x_{2}))((x_{1}x_{2})(y_{1}y_{2}))$.
Up to dropping the $J_1$ factor, and interchanging the commuting pair $x_{2}y_{1}$, this can be read off from product of the bottom row of the square.

This provides a more succinct contextuality witness than the usual parity proof, which amounts to taking the product of all the rows and columns.

A similar treatment can be given of the Mermin star \cite{mermin1990simple}.
\end{example}

\subsection{Comparison with other Heisenberg groups}
We can now see why taking the more standard Heisenberg group construction, defined exactly as for $H(\mu)$, but using $\mu$ rather than $\mutr$, would not be suitable for our purposes. Let us denote the construction using $\mu$ rather than $\mutr$ by $H^+(\mu)$. By Proposition~\ref{commutatorprop}, the commutator in $H(\mu)$ is $\mu$, which means that we can have commuting products $gh = hg$ with non-zero but equal phase factors, which is clearly essential for contextual words to exist. By contrast, the commutator in $H^+(\mu)$ is easily seen to be $2\mu$, which means in $Z_2$ that \emph{all products commute}, while in odd orders, no products commute.

\section{No state-independent contextuality in odd characteristics}
\label{odd}
We shall now show that contextual words can only exist over $\ZZ_d$ if $d$ is even.
Moreover, we shall explicitly describe the non-contextual value assignments which exist when $d$ is odd.

In order to prove these results, we will analyze the structure of inversions in bracketed words.

\textbf{Notation} In this section, we will deal exclusively with non-empty words over the generators, $w \in \XG^+$.
We will also write formal sums in variables $v_{x,y}$ to stand in for values of the commutator matrix $\mu(x,y)$.
We will use the following notation. If $S = \{ \lambda_i \}_{i \in I}$ is a family of inversions, then $\sum S \, := \, \sum_{i \in I} v_{x_i,y_i}$, where $\lambda_i$ is an inversion between $x_i$ and $y_i$, $x_i > y_i$.

Given a word $s$, we write $\Inv(s)$ for the set of inversions in $s$. Given words $s$, $t$, $\Inv(s,t)$ is the set of inversions between $s$ and $t$, \ie~the set of all $(w_1, x, w_2, y)$ such that $w_1x$ is a prefix of $s$, $w_2y$ is a prefix of $t$, and $x > y$.

The following is immediate.
\begin{lemma}
\label{invlemm}
For all words $s$, $t$, 
\[ \sum {\Inv(st)} \; = \; \sum {\Inv(s)} + \sum {\Inv(t)} + \sum {\Inv(s,t)} . \]
\end{lemma}
In this notation, the equation forcing the global phase factor for a word $s$ to be $k$ is  
\[ \sum \Inv(s) \; = \; k . \]
Also, given words $s$, $t$, we define the ``formal commutator'' of $s$ and $t$ to be
\[ \lcomm s, t \rcomm \; := \; \sum \Inv(s,t) - \sum \Inv(t,s) . \]
The equation forcing $s$ and $t$ to commute is 
$\sum \Inv(st) - \sum \Inv(ts) = 0$. By the previous lemma, this is equivalently written as  $\lcomm s,t \rcomm = 0$.

\begin{lemma}
\label{revlemm}
Let $s^{\dagger}$ be the reverse of a  word $s$.  Then $\sum \Inv(s,t) \, = \, \sum \Inv(s^{\dagger}, t^{\dagger})$.
\end{lemma}
\begin{proof}
Note that $s^{\dagger}$  defines the same multiset of occurrences of generators as $s$, so there will be a bijection between the inversions in $\Inv(s,t)$ and those in $\Inv(s^{\dagger}, t^{\dagger})$, inducing the same multiset of variables $v_{x,y}$.
\end{proof}

We now consider bracketings of  words.
Given a bracketing $\beta$, we define the multiset $\Phi(\beta)$ by 
\[ \Phi(x)  = \es, \qquad \Phi(\beta_1,\beta_2) = \mlb (\bd(\beta_1), \bd(\beta_2)) \mrb  \uplus \Phi(\beta_1) \uplus \Phi(\beta_2) . \]
Given a word $s$ with bracketing $\beta$, we  can write $\Phi(\beta)$ as a family $\{ (s_i,t_i) \}_{i \in I}$ of adjacent subwords of $s$ corresponding to subexpressions of the full bracketing. 
\begin{lemma}
\label{commlemm}
With notation as above, let $s^{\dagger}$ be the reverse of $s$. Then
\begin{equation}
\label{eqrev}
\sum_{i \in I} \, \lcomm s_i,t_i \rcomm \;\; = \;\; \sum \Inv(s) \; - \; \sum \Inv(s^{\dagger}) . 
\end{equation}
\end{lemma}
\begin{proof}
By induction on the length of $s$. If $s = x$, then the sums on both sides of~(\ref{eqrev}) are empty, and we have the equation $0 = 0$.

In the inductive case, suppose the top-level bracketing of $s$ is $s = uv$. We can write the bracketings of $u$ and $v$ as families $\{ (u_j, u'_j) \}_{j \in J}$, $\{ (v_k, v'_k) \}_{k \in K}$.
Then, applying the induction hypothesis:
\[ \begin{array}{lcl}
\sum_{i \in I} \, \lcomm s_i,t_i \rcomm & = & \lcomm u,v \rcomm \; + \; \sum_j \lcomm u_j, u'_j \rcomm  \; + \; \sum_k \lcomm v_k, v'_k \rcomm \\
& = & (\sum \Inv(u,v) - \sum \Inv(v,u)) \; + \; (\sum \Inv(u) - \sum \Inv(u^{\dagger})) \; + \; (\sum \Inv(v) - \sum \Inv(v^{\dagger})) .
\end{array}
\]
By Lemma~\ref{invlemm}, $\sum \Inv(s) \, = \,\sum \Inv(uv) \, = \, \sum \Inv(u, v) + \sum \Inv(u) + \sum \Inv(v)$.
Since $s^{\dagger} = v^{\dagger}u^{\dagger}$, applying Lemma~\ref{invlemm} again yields $\sum \Inv(s^{\dagger})  \; = \; \sum \Inv(v^{\dagger}, u^{\dagger}) + \sum \Inv(v^{\dagger}) + \sum \Inv(u^{\dagger})$.
Applying Lemma~(\ref{revlemm}) and rearranging terms yields~(\ref{eqrev}).
\end{proof}

\begin{lemma}
\label{dlemm}
Let $w$ be a word in which each generator $x$ occurs $n_x$ times, modulo $d$. Then $$\sum \Inv(w) + \sum \Inv(w^{\dagger})=\sum_{x<y} n_{y} n_{x} v_{yx}$$ 
In particular, when each generator occurs a multiple of $d$ times, we have $\sum \Inv(w) \, = \, - \sum \Inv(w^{\dagger})$.
\end{lemma}
\begin{proof}
In order to count the number of occurrences of the variable $v_{yx}$, consider each occurrence of $y$ within $w$. For the $i^{th}$ occurrence we can write $w = u_iyv_i$ and for each $x$ such that $x < y$, the $m_i$ occurrences of $x$ in $v_i$ will yield $m_i$ inversions in $w$, while the occurrences of $x$ in $u_i$ will yield $n_i$ inversions in $w^{\dagger}$. Thus, the total multiplicity of $v_{x,y}$ in $\sum \Inv(w)$ will be $\sum_i m_i$, where $i$ ranges over occurrences of $y$ in $w$. Similarly, the total multiplicity of $v_{x,y}$ in $\sum \Inv(w^{\dagger})$ will be  $\sum_i n_i$. Since for each $i$, $m_i + n_i = n_x$ we will have the overall multiplicity $$n_{yx}=\sum_{i=1}^{n_y} m_i+ \sum_{i=1}^{n_y} n_i = \sum_{i=1}^{n_y} n_x = n_y  n_x.$$ When each $n_x$ is a multiple of $d$ we therefore have $\sum \Inv(w) + \sum \Inv(w^{\dagger})=0 \pmod d$.
\end{proof}

\begin{theorem}
\label{cwordth}
If $(w, \beta, k)$ is a contextual word over $\ZZ_d$, then $d$ is even.
\end{theorem}
\begin{proof}
Since $(w, \beta, k)$ is contextual, we have $\lcomm s,t \rcomm = 0$ for all bracketed subexpressions $(s,t)$ in $\beta$. Hence summing over all such subexpressions yields $\sum_i \lcomm s_i, t_i \rcomm = 0$.
By Lemma~\ref{commlemm}, this implies that $\sum \Inv(w) \, - \, \sum \Inv(w^{\dagger}) = 0$. Applying Lemma~\ref{dlemm} yields $2 \sum \Inv(w) = 0$. The contextuality of $(w, \beta, k)$ forces $\sum \Inv(w) = k$, where $k \neq 0$. We can only have a non-zero solution of $2k = 0 \pmod d$ if $d$ is even.
\end{proof}

\begin{theorem}
\label{commonphaselemm}
If $w_1$ and $w_2$ are words in $\XG^+$ formed out of commuting products and which have the same multiset of generators, modulo $d$, then if $d$ is odd their overall commutation factors are equal.
\end{theorem}
\begin{proof}
Since $w_1$ and $w_2$ are formed out of commuting products, each of the formal commutators corresponding to the sub-expressions of $w_1$ and $w_2$ is equal to zero, hence $\sum \Inv(w_i)=\sum \Inv(w_i^{\dagger})$. From Lemma \ref{dlemm} it follows that $$2\sum \Inv(w_i)=\sum_{x<y} n^i_{y}n^i_{x}\cdot v_{yx}$$ 
The right hand side of this equation is the same for $w_1$  and $w_2$, since the number $n^i_x$ of occurrences of each generator is equal modulo $d$ in the two words.
Since  $d$ is odd, the equation $2x = k \pmod d$ has a unique solution for any $k \in \ZZ_d$, and  $\sum \Inv(w_1)=\sum \Inv(w_2)$.
\end{proof}

\begin{theorem}
Let $\mu$ be a commutator matrix over $\ZZ_d$. If $d$ is odd, there is a non-contextual value assignment $\val : \CM(\mu) \to \ZZ_d$. 
\end{theorem}
\begin{proof}
By Theorem~\ref{cwordth}, if $d$ is odd, there are no contextual words over $\CG(\mu)$. By Theorem~\ref{cwsicth}, there is a non-contextual value assignment $\val : \CM(\mu) \to \ZZ_d$. 
\end{proof}


\section{Contextuality in even characteristics}

We now show that contextual words exist in abundance in even characteristics.
Firstly, we characterize the circumstances under which non-contextual value assignments \emph{do} arise.


Given $S \subseteq \CM(\mu)$, we define $\Zen(S) := \{ a \in S \mid \forall b \in S. \, a \odot b \}$. A \emph{graph} will always mean a reflexive undirected graph, \ie a set of vertices with a reflexive, symmetric relation.

A \emph{cluster graph} is a coproduct (disjoint union) of complete graphs. Equivalently, it is a graph in which the adjacency relation is transitive, so that the maximal cliques are the equivalence classes, and hence disjoint, with no adjacencies between them.

We will show that,  if $(\CM(\mu) \setminus \Zen(\CM(\mu)), \odot)$ is a cluster graph, every empirical model over $\CM(\mu)$ has global sections, which are exactly non-contextual value assignments.

We briefly review what we need of empirical models; for further details, see \cite{abramsky2011sheaf,DBLP:conf/csl/AbramskyBKLM15}.
A maximal clique over the graph $(\CM(\mu), \odot)$ is a (total) commutative sub-monoid of $\CM(\mu)$: closure under products is implied by maximality. Moreover, it contains $\Zen(\CM(\mu))$.
Let $\MM$ be the set of maximal cliques. Note that the union of this family is $\CM(\mu)$.

A (possibilistic) empirical model over $\CM(\mu)$ assigns to each $C \in \MM$ a non-empty set of homomorphisms $s : C \to \ZZ_d$ which split the inclusion $\ZZ_d \hookrightarrow C$. We write $\{ e_C \}_{C \in \MM}$ for this family of sets of homomorphisms. The family is moreover required to satisfy the following \emph{local consistency} property: for all $C, C' \in \MM$, $e_C |_{C \cap C'} = e_{C'} |_{C \cap C'}$, where e.g
\[ e_C |_{C \cap C'} := \{ s |_{C\cap C'} \mid s \in e_C \} . \]
We say that such an empirical model is non-contextual (in the sense of not strongly contextual \cite{abramsky2011sheaf}) if there exists a \emph{global section}: a homomorphism $s : \CM(\mu) \to \ZZ_d$ such that $s |_C \in e_C$ for all $C \in \MM$.
Such a global section is necessarily a left splitting, and hence a non-contextual value assignment for $\CM(\mu)$.

\begin{theorem}
\label{algclusterprop}
 If $(\CM(\mu) \setminus \Zen(\CM(\mu)), \odot)$ is a cluster graph, then every empirical model over $\CM(\mu)$ is non-contextual.
\end{theorem}
\begin{proof}
Let $N := \CM(\mu) \setminus \Zen(\CM(\mu)), Z := \Zen(\CM(\mu))$. Each maximal clique of $(\CM(\mu), \odot)$ is of the form $C \sqcup Z$, where $C$ is a maximal clique of $(N, \odot)$. 
Let $e$ be an empirical model, and consider $s \in e_{C \sqcup Z}$ for $C \sqcup Z \in \MM$. We can write $s = [s_C, s_Z] : C \sqcup Z \to \ZZ_d$.  By the local consistency property for $e$, for any $C' \neq C$ maximal in $(N, \odot)$, there is $s_{C'} : C' \to \ZZ_d$ such that $s' = [s_{C'}, s_{Z}] : C' \sqcup Z \to \ZZ_d \in e_{C' \sqcup Z}$. Morever, as $C'$, $C''$ range over maximal cliques of  $(N, \odot)$, since $C' \cap C'' = \es$, $s' = [s_{C'}, s_{Z}]$ is compatible with $s'' = [s_{C''}, s_{Z}]$, \ie $s' |_{Z} = s_Z = s'' |_{Z}$. Thus we obtain a pairwise compatible family of sections $\{ [s_C, s_Z] \}_{C \sqcup Z \in \MM}$. 

Since $\MM$ covers $\CM(\mu)$, this family determines a unique function $s : \CM(\mu) \to \ZZ_d$. We must check the homomorphism condition. This holds because whenever $g \odot h$, $\{ g, h \} \subseteq C$ for some $C \in \MM$, hence $s(gh) = s_C(gh) = s_C(g) + s_C(h) = s(g) + s(h)$.
\end{proof}
Note that in the last part of the argument, we were verifying the sheaf property for the cover $\MM$ over the presheaf of left splittings on cliques in $(\CM(\mu), \odot)$.

One remaining question is whether empirical models over $\CM(\mu)$  actually exist.\footnote{The issue is whether we can have a non-empty model satisfying the local consistency conditions.} We shall discuss unitary representations of commutation groups in the next section. Given a quantum realization of the associated measurements, we can always obtain an empirical model by applying any quantum state.

\subsection{Positive results}
We now assume that $(\CM(\mu) \setminus \Zen(\CM(\mu)), \odot)$ is \emph{not} a cluster graph, which means that there are elements $a,b,c$ such that $a \odot b$, $a \odot c$, but not $b \odot c$. 
Since $a$ is not in $\Zen(\CM(\mu))$, there must be some $d$ such that not $a \odot d$. Allowing for the various possibilities for commutativity of $d$ with $b$ and $c$, up to relabelling this gives us the
following three compatibility graphs \cite{kirby2021variational}:
\begin{equation}
\label{graphseq}
\begin{tikzcd}
a \ar[r, dash] \ar[d, dash] & b \\
c & d
\end{tikzcd}
\qquad
\begin{tikzcd}
a \ar[r, dash] \ar[d, dash] & b \\
c \ar[r, dash] & d
\end{tikzcd}
\qquad
\begin{tikzcd}
a \ar[r, dash] \ar[d, dash] & b  \ar[d, dash] \\
c \ar[r, dash] & d
\end{tikzcd}
\end{equation}

\subsubsection{The $\ZZ_2$ case}

In the case where $\mu$ is a commutator matrix over $\ZZ_2$, we can give a definitive characterisation of contextuality in $\CM(\mu)$. This follows similar lines to \cite{kirby2021variational}, in the general setting of commutation groups.
\begin{theorem}
If $\mu$ is a commutator matrix over $\ZZ_2$, then the following are equivalent:
\begin{enumerate}
\item $\CM(\mu)$ is contextual.
\item There are contextual words over $\CM(\mu)$.
\item The graph $(\CM(\mu),\odot)$ contains one of the graphs in~(\ref{graphseq}) as an induced sub-graph.
\end{enumerate}
\end{theorem}
\begin{proof}
The implication $(2) \Rightarrow (1)$ follows from Theorem~\ref{cwsicth}. By contraposition, $(1) \Rightarrow (3)$ follows from Theorem~\ref{algclusterprop}.
Now assume $(3)$. If $a$, $b$, $c$, $d$ are generators, the matrices $\mu_i$ given in Example~\ref{KLex} correspond to the graphs in~(\ref{graphseq}), and the corresponding contextual words given in the Example show that $(2)$ holds. Otherwise, these elements arise as commuting products, each of which can be described by a suitably bracketed word. If any of these words has global phase factor $1$, they are already contextual words. Otherwise, we can substitute them into the words given in Example~\ref{KLex} to obtain contextual words.
\end{proof}

\subsubsection{Beyond $\ZZ_2$: padding, splitting and variable changes}

We can transfer contextual words from $\ZZ_2$ to $\ZZ_{2k}$, using the embedding $\ZZ_2 \rightarrowtail \ZZ_{2k}$ which sends $1$ to $k$, which can be applied to a commutator matrix over $\ZZ_2$ to produce one over $\ZZ_{2k}$.
If we take any of the contextual words $w$ from Example~\ref{KLex}, we can then perform a simple padding construction.
 We append $a^{2k-2}b^{2k-2}c^{2k-2}d^{2k-2}$ to $w$, and this produces a contextual word over $\ZZ_{2k}$. 

Can we construct contextual words over $\ZZ_{2k}$ using matrix values other than $0$ and $k$? By Theorem~\ref{cwordth}, the global phase factor for a contextual word over $\ZZ_{2k}$ must be $k$, but we may use other values from $\ZZ_{2k}$ in constructing the word.
We can use a splitting construction to achieve this. We illustrate the idea with a simple example over $\ZZ_4$. Given the contextual word $((ab)(dc))((ca)(bd))$ from Example~\ref{KLex}, we split the generator $a$ into $a_1$ and $a_2$. We can use the commutator matrix
$$\mu= \begin{bmatrix}
0 & 0 & 1 & 1 & 1 \\
0 & 0 & 3 & 3 & 1 \\
3 &1 & 0& 2 & 0 \\
3 & 1 & 2 & 0 & 0 \\
3 & 3 & 0 & 0 & 0
\end{bmatrix}
$$
and obtain the contextual word  $[((a_1a_2)b)(cd)][((a_1a_2)c)(bd)] [a_1^{2}a_2^2b^2c^2d^2]$, using also the padding construction described previously.

\subsubsection{Classification of contextuality for matrices in Darboux normal form}

Our overall aim is to give a complete classification of which commutation groups $\CG(\mu)$ admit contextual words.
We shall achieve this for matrices $\mu$ in \emph{Darboux normal form}, \ie  whose only non-zero entries occur in block $2\times 2$ matrices on the main diagonal of the form
$$
\begin{array}{cc}
    0 & \lambda \\
    -\lambda & 0
\end{array}
$$

Given a matrix $\mu$  over $\ZZ_{2n}$, which is in Darboux normal form, it is possible to decide whether it supports contextual words by considering the parity, relative to $n$, of the non-zero entries of $\mu$. By relative parity, we mean whether the power with which $2$ appears in the prime factor decomposition of $n$ is lower than the power with which it appears in the prime factor decomposition of each of the non-zero entries. Thus, if $n=m'\times 2^m$ and $\lambda=l'\times 2^l$, where $m'$ and $l'$ are both odd integers, we say that $\lambda$ is even relative to $n$ if $l>m$ and, conversely, that $\lambda$ is odd relative to $n$ if $l\leq m$.

\begin{theorem}
If $\mu$ is in Darboux normal form with entries in $\ZZ_{2n}$, then there is a contextual word over $\CG(\mu)$ if and only if  there are two non-zero entries above the main diagonal, $\lambda_1=l'_1\times 2^{l_1}$ and $\lambda_2=l'_2\times 2^{l_2}$ which are both odd relative to $n=m'\times 2^m$.
\end{theorem}

\begin{proof}
If we can find two non-zero entries above the main diagonal, $\lambda_1$ and $\lambda_2$ which are both odd relative to $n$ and we denote their corresponding variables by $a, b, c, d$, so that $\mu(a,b)=\lambda_1$ and $\mu(c,d)=\lambda_2$ then we can form the contextual word

$$((\underbrace{a\ldots a}_{k_a} \ \underbrace{c\ldots c}_{k_c})(bd))((\underbrace{a\ldots a}_{k_a} d)(b\underbrace{c\ldots c}_{k_c}))(\underbrace{a\ldots a}_{2n-2k_a})(\underbrace{b\ldots b}_{2n-2})(\underbrace{c\ldots c}_{2n-2k_c})(\underbrace{d\ldots d}_{2n-2})$$
where $k_a=m'\times 2^{m-l_1}$ and $k_c=m'\times 2^{m-l_2}$. It is straightforward to check that all the brackets commute and that the non-zero contributions to the overall commutation factor come from sliding $k_a$ occurences of the variable $a$ leftwise over the first occurence of the variable $b$, giving a total contribution equal to $$k_a\times\lambda_1 = (m'\times 2^{m-l_1})\times (l'_1\times 2^{l_1})=n\times l'_1$$ which is equivalent to $n$ modulo $2n$, since $l'_1$ is odd. Note that, since $2k_i\times\lambda_i\equiv0\ (mod\ 2n)$, there are no other non-zero contributions to the overall commutation factor.

On the other hand, if all the non-zero entries of $\mu$ are even relative to $n$ then we will not be able to get any word $w$ with a non-zero commutation factor since, as we have shown in Section \ref{odd}, the commutation factor of $w$ must satisfy the equation $$2\sum {\Inv(w)}=0$$ 

In $\ZZ_{2n}$ the only non-zero solution to this equation is $\sum {\Inv(w)}=n=m'\times 2^m$ and since all the commutation variables in the matrix $\mu$ have a power of $2$ greater than $m$ in their prime factor decomposition, any linear combination of these values will also have a power of $2$ greater than $m$ in its prime factor decomposition, and so will not yield a commutation factor equivalent to $n$ modulo ${2n}$.

Finally, we can show that if only one non-zero entry is odd relative to $n$, while all the others are relatively even, then contextual words cannot exist. For simplicity of notation, we will show this for $4$ generators but the proof, which is essentially a parity argument, works equally well for any number of generators. Let $n=m'\times 2^m$, as before, and let $a$ and $b$ denote the two variables whose corresponding entry in the commutation matrix is $\lambda=l'\times 2^{m-l}$ for some odd $l'$ and $l\geq 0$. By assumption, all other entries are of the form $\lambda_i=l'_i\times 2^{m+1+l_i}$ for some odd $l'_i$ and $l_i\geq 0$. If, as we assumed at the beginning, we have only one such relatively even $\lambda_i$ corresponding to two extra variables $c$ and $d$, then any bracketed subexpression of the form
$$w=(a^{L_a}b^{L_b}c^{L_c}d^{L_d})(a^{R_a}b^{R_b}c^{R_c}d^{R_d})$$

commutes if and only if 

\begin{equation*}
    \begin{aligned}
\lambda(L_aR_b-R_aL_b) &+ \lambda_i(L_cR_d-R_cL_d) \equiv 0\ (mod\ 2n)\\
\Leftrightarrow \ 2^{m-l}\times l' (L_aR_b-R_aL_b) &+2^{m+1+l_i}\times l'_i(L_cR_d-R_cL_d) =2 N \times 2^{m}\times m'
    \end{aligned}
\end{equation*}

Since the right hand side is a multiple of $2^{m+1}$ and the terms on the left hand side coming from the variables with a relatively even commutation factor are also multiples of $2^{m+1}$, as $l'$ is odd, it follows that 
$$(L_aR_b-R_aL_b) \, \dvs \, 2^{l+1}$$
If we were able to show that $R_aL_b$ is a multiple of $2^{l+1}$ then the overall commutation factor would be
$$\lambda R_aL_b + \lambda_i R_cL_d=2^{m-l}\times l' \times R_aL_b + 2^{m+1+l_i}\times l'_i R_cL_d$$
and it would also be a multiple of $2^{m+1}$. Since the only possible non-zero value for the overall commutation factor is $n$ modulo $2n$, a sum of terms which are divisible by $2^{m+1}$ cannot yield a non-zero overall commutation factor.

To show that $R_aL_b$ is a multiple of $2^{l+1}$, we start by proving by induction on the length of the word $w$ that, if $w$ is formed out of commuting products, then $$T_aT_b \, \dvs \, 2^{l+1}$$ where $T_a$ and $T_b$ denote the multiplicities with which the variables $a$ and $b$ appear in $w$.

If $w$ has length $1$ then at least one of $T_a$ or $T_b$ is equal to zero and the statement holds trivially. Assume the statement holds for words of length at most $W$ and let $w$ be a word of length $W+1$ of the form  
$$w=(a^{L_a}b^{L_b}c^{L_c}d^{L_d})(a^{R_a}b^{R_b}c^{R_c}d^{R_d})$$
By the inductive hypothesis, since $w$ is formed out of two well-bracketed words of strictly smaller length, we must have that $L_aL_b \, \dvs \, 2^{l+1}$ and $R_aR_b\, \dvs \, 2^{l+1}$. Therefore $$L_aL_bR_aR_b\, \dvs \, 2^{2l+2}$$ which, by the pigeon-hole principle implies that either $L_aR_b$ or $R_aL_b$ must be a multiple of $2^{l+1}$.

We have argued previously that, since the words within the two brackets commute we must have $$(L_aR_b-R_aL_b) \, \dvs \, 2^{l+1}$$
But since at least one of the two terms is a multiple of $2^{l+1}$ it is in fact necessary that both of terms are divisible by $2^{l+1}$.

This both completes the inductive proof, as $T_a=L_a+R_a$ and $T_b=L_b+R_b$ and so $T_aT_b$ expands as a sum of four terms which are all divisible by $2^{l+1}$:
$$T_aT_B= L_aL_b+L_aR_b+R_aL_b+R_aR_b$$
and also allows us to conclude that it is not possible to obtain a non-zero contribution to the overall commutation factor from the $a$ and $b$ variables if these are the only two variables with a relatively odd commutation factor.

\end{proof}

\subsubsection{Reduction to Darboux normal form}

Every commutator matrix can be reduced to one in Darboux normal form. This is standard over a field, but less obvious over $\ZZ_d$ for arbitrary $d$, so we include a proof.

Note that since $\mu$ plays the role of a bilinear form, if we wish to perform a change of basis preserving this form, we must perform the corresponding row and column operations on the matrix $\mu$. These operations are encoded by an invertible base change matrix $U$; the resulting matrix $U^T\mu U$ is said to be \emph{cogredient} to $\mu$. 

\begin{lemma}\label{clearing}
    Every commutation matrix $\mu$ is cogredient to a skew-symmetric matrix $\mu_d$  whose only non-zero entries occur directly above and below the main diagonal. We call this the standard form of the commutation matrix.
\end{lemma}
\begin{proof}

We start by noting that the `swapping' matrix $U_{i,j}$ which is obtained by swapping the $i^{th}$ and $j^{th}$ columns of the identity matrix is self-inverse over $\ZZ_{2n}$ and $$U^T_{i,j}\mu U_{i,j}$$ is the commutation matrix obtained from $\mu$ by first swapping the $i^{th}$ and $j^{th}$ rows of $\mu$ and then the $i^{th}$ and $j^{th}$ columns. 

Similarly, the `adding' matrix $V^{\alpha}_{i, j}$ which is obtained by adding to the $i^{th}$ column of the identity matrix $\alpha$ times the $j^{th}$ column is invertible, its inverse being $V^{-\alpha}_{i, j}$. Hence the matrix $$V^{\alpha T}_{i,j}\mu V^{\alpha}_{i, j}$$ is cogredient with $\mu$ and has the corresponding effect of adding $\alpha$ times the $j^{th}$ row/column of $\mu$ to its $i^{th}$ row/column. 

Using these two types of cogredient operations it is possible, using Euclid's algorithm to change $\mu$ into a cogredient matrix $\mu_d$ with the desired property. 
\\

The first step is to consider the $n^{th}$ row of the matrix $\mu$:
$$
\begin{array}{ccccc}
\ddots &\vdots &\vdots &\vdots &\vdots \\
    \ldots&0 & * & *& -c \\
    \ldots&* & 0 & * & -b\\
   \ldots& * & * & 0 & -a  \\
  \ldots&  c & b & a & 0    
\end{array}
$$
If it has $k$ non-zero entries, we can use suitable swapping $U_{i,j}$ matrices to bring those entries to the right of all zero entries, ordered ascendingly. Then if $a=\mu(n,n-1)>\mu(n,n-2)=b$ we can use  $V^{\alpha}_{n-1,(n-2)}$ and $U_{n-1,n-2}$ matrices to perform Euclid's algorithm on the bottom entries of these penultimate two columns, resulting in a cogredient matrix $\mu_1$ with $\mu_1(n,n-2)=0$ and $\mu_1(n,n-1)=gcd(a,b)$. For example, if the first step of the algorithm gives the decomposition $a=bq_1+r_1$ then the matrix
$$\mu'=U_{n-1,n-2}^TV^{-q_1 T}_{n-1, n-2}\mu V^{-q_1}_{n-1, n-2}U_{n-1,n-2}$$ 
will have $\mu'(n,n-2)=r_1$ and $\mu'(n,n-1)=b$. If $r_1$ is non-zero, we can continue iterating the next steps of the algorithm, until eventually we reach $\mu_1$.

The next step is to consider $c=\mu_1(n,n-3)$ and use suitable $U_{n-3,n-1}$ and $V^{\alpha}_{n-1,n-3}$ matrices to again perform Euclid's algorithm, resulting in a matrix $\mu_2$ for which $\mu_2(n,n-2)=\mu_2(n,n-3)=0$ and $\mu_2(n,n-1)=gcd(a,b,c)$. And we proceed to eliminate all the $k$ next non-zero entries of the last row, thus leaving $$\mu_{k-1}(n,n-1)=gcd(\mu(n,n-1),\mu(n,n-2),\ldots,\mu(n,n-k))$$ as the only non-zero entry on the $n^{th}$ row. And since cogredient operations result in skew-symmetric matrices, the only non-zero entry on the $n^{th}$ column will also be the one above the main diagonal.

We can repeat these steps to clear out the $t$ non-zero entries on row $n-1$, which are to the left of the $(n-1,n-2)$ position. This results in some matrix $\mu_{t-1}$ whose only nonzero entries on row $n-1$ are $\mu_{t-1}(n-1,n)=\mu_{k-1}(n-1,n)$ and
\[ \mu_{t-1}(n-1,n-2)=gcd(\mu_{k-1}(n-1,n-2),\mu_{k-1}(n-1,n-3),\ldots,\mu_{k-1}(n-1,n-t)). \]
We proceed similarly with the remaining rows, eventually resulting in a matrix $\mu_d$ in standard form.
\end{proof}

\begin{theorem}
    Every commutation matrix $\mu$ is cogredient to a matrix $\mu_D$ in Darboux normal form, whose only non-zero entries occur in block $2\times 2$ matrices on the main diagonal of the type
$$
\begin{array}{cc}
    0 & \lambda_i \\
    -\lambda_i & 0
\end{array}
$$
\end{theorem}

\begin{proof}

From Lemma \ref{clearing}, we know that $\mu$ is cogredient to a matrix $\mu_d$ in standard form. We describe an algorithm which, given a $4\times 4$ diagonal block of $\mu_d$ of the type 
$$
\begin{array}{cccc}
    0 & a & 0& 0 \\
    -a & 0 & b & 0\\
    0 & -b & 0 & c \\
    0 & 0 & -c & 0    
\end{array}
$$
performs cogredient operations on $\mu_d$ to produce a block in Darboux normal form of the type 
$$
\begin{array}{cccc}
    0 & \lambda_1 & 0& 0 \\
    -\lambda_1 & 0 & 0 & 0\\
    0 & 0 & 0 & \lambda_2 \\
    0 & 0 & -\lambda_2 & 0    
\end{array}
$$
Assume without loss of generality that the $4\times 4$ block is in the top left-hand corner of $\mu_d$. If $b=0$ the block already has the desired format. If $c$ is equal to zero, we can use `swapping' $U_{1,3}$ and suitable `adding' $V^{\alpha}_{1,3}$ matrices to perform Euclid's algorithm on the non-zero entries in the first and third columns, resulting in a block of the desired format, with $\lambda_1=gcd(a,b)$ and $\lambda_2=0$, and the same type of procedure can be used when $a=0$. 

In the remaining case, when all entries are non-zero, we have to distinguish two scenarios: first, if $aq=b$ then $V^{qT}_{3,1}\mu_d V^{q}_{3,1}$ has the top left-hand block in Darboux normal form with $\lambda_1=a$ and $\lambda_2=c$. 

Otherwise, performing Euclid's algorithm as above will initially result in a block with non-zero entries away from the main diagonal:

$$
\begin{array}{cccc}
    0 & gcd(a,b) & 0& y \\
    -gcd(a,b) & 0 & 0 & 0\\
    0 & 0 & 0 & x \\
    -y & 0 & -x & 0    
\end{array}
$$
We now make a slight modification to the procedure in Lemma \ref{clearing} in order to bring this block matrix back to standard form: instead of applying Euclid's algorithm to reduce the entries on the last row and column, we use it to reduce the entries on the first row and column. The resulting matrix will be of the form 
$$
\begin{array}{cccc}
    0 & gcd(a,b,y) & 0& 0 \\
    -gcd(a,b,y) & 0 & y' & 0\\
    0 & -y' & 0 & x' \\
    0 & 0 & -x' & 0    
\end{array}
$$
At this point we can repeat the steps outlined so far until we eventually bring the block to Darboux normal form.  Since we started with the assumption that $a$ is not a factor of $b$, the greatest common divisor of $a,b$ and $y$ must be strictly less than $a$ in the divisibility order, so eventually the process will terminate. 

\end{proof}

\textbf{Discussion} $\,$ There is an important caveat to this result. We have defined contextuality for $\CG(\mu)$ in terms of $\CM(\mu)$, which is defined relative to the set of generators $\XG$. When transforming $\mu$ to $\mu'$ in Darboux normal form, we will have $\CG(\mu) \cong H(\mu) \cong H(\mu') \cong \CG(\mu')$, but this transformation will not preserve the generators. In particular, the new generators corresponding to the transformed basis for $\mu'$ may correspond to words in the old generators which cannot be formed from commuting products. This means that a contextual word over $\CG(\mu')$ may not correspond to one over $\CG(\mu)$.

\section {Unitary representation}

Since $\CG(\mu)$ is a finite group, it has unitary representations. Indeed, every linear representation is equivalent to a unitary one.
We wish to have unitary representations which faithfully preserve the internal $\ZZ_d$-action.

We use the qudit Hilbert space $\HH_d := \Complex^d$, with basis vectors $\ket{k}$ labelled by elements of $\ZZ_d$. The tensor product of $n$ copies of this space, $\HH_{n,d}$, has basis vectors $\ket{\kv}$ labelled by $\kv \in \ZZ_d^n$. We write $\UG(\HH_{n,d})$ for the unitary group on $\HH_{n,d}$. 
The centre of this group is isomorphic to the circle group, $\UG(1) := \{ z \in \Complex \mid |z| = 1 \}$. For each $d \geq 2$, this contains the cyclic subgroup of the $d$'th complex roots of unity.
We write $\omega := e^{\frac{2\pi i }{d}}$ for the primitive $d$'th root of unity. The map $k \mapsto \omega^k$ is an isomorphism from $\ZZ_d$ to the multiplicative group of $d$'th complex roots of unity.

Given a commutator matrix $\mu \in \cm(n, \ZZ_d)$, we shall define a representation $\rho : H(\mu) \to \UG(\HH_{n,d})$:
$\rho (k, \kv) \, \ket{\lv} \; = \; \omega^{k + \mutr(\kv,\lv)} \, \ket{\lv + \kv}$.
\begin{proposition}
For each $(k, \kv) \in H(\mu)$, $\rho (k, \kv)$ is a well-defined unitary operation. Moreover, $\rho$ is an injective group homomorphism which preserves scalars, \ie $\rho (k,0) = \omega^k \mathbb{1}$.
\end{proposition}
\begin{proof}
The verification that $\rho$ is a homomorphism amounts to showing that 
\[ \rho(k, \kv) \circ \rho(k', \kv') \, \ket{\lv} \; = \; \rho(k + k' + \mutr(\kv, \kv'), \kv + \kv') \, \ket{\lv} \]
which reduces to 
\[ \mutr(\kv, \kv') + \mutr(\kv + \kv', \lv) \; = \; \mutr(\kv', \lv) + \mutr(\kv, \kv' + \lv) \]
which follows from bilinearity.
\end{proof}

\subsection{Representation in Pauli groups}
The generalized Pauli groups $\Pnd$ are the subgroups of  $\UG(\HH_{n,d})$  generated by the $X$ and $Z$ operations. These operations are defined on $\HH_d$ by $X \, \ket{k} = \ket{k+1}$, and $Z \, \ket{k} = \omega^k \ket{k}$. 
These are the Sylvester ``shift''  and ``clock'' matrices \cite{sylvester1883quaternions}, and can be seen as discrete versions of  position and momentum operators.
Note that they satisfy the basic commutation relation $ZX = \omega XZ$.
They are then extended to $\HH_{n,d}$ as $X_i := \underbrace{I \otimes \cdots \otimes I}_{i-1} \otimes X \otimes \underbrace{I \otimes \cdots \otimes I}_{n-i}$, and similarly for $Z_i$, $i=1, \ldots ,n$. Note that the commutator matrix for the generators $X_i$, $Z_i$ is in Darboux normal form: the only non-zero entries are $\mu(Z_i,X_i) = 1$, $\mu(X_i,Z_i) = -1 \pmod d$.

\begin{proposition}
The image of $H(\mu)$ under $\rho$ is a subgroup of $\Pnd$.
\end{proposition}
\begin{proof}
Given an element $e_i$ of the standard basis of $\ZZ_d^n$, we have $\rho(e_i) = X_i \prod_{j=1}^n Z_j^{\mutr_{i,j}}$, which can be verified by a simple computation:
\[ \begin{array}{lcl}
X_i \prod_{j=1}^n Z_j^{\mutr_{i,j}} \, \ket{k_1} \otimes \cdots \otimes \ket{k_n} & = & X_i (\omega^{\mutr_{i,1}k_1} \ket{k_1} \otimes \cdots \otimes \omega^{\mutr_{i,n}k_n} \ket{k_n}) \\
& = & X_i (\prod_j \omega^{\mu_{i,j}k_j}) \, \ket{\kv} \\
& = & \omega^{\sum_j \mutr_{i,j} k_j} X_i \, \ket{\kv} \\
& = & \omega^{\mutr(e_i,\kv)} \, \ket{\kv + e_i} \\
& = & \rho(e_i) \, \ket{\kv} .
\end{array}
\]
Since $H(\mu)$ is generated by the $e_i$ and the scalar $(1,0)$, this yields the result.
\end{proof}
This result shows the universality of the Pauli operations for expressing discrete commutation relations. At the same time, the structural tools made available by the presentations of commutation groups allow for a fine-grained analysis of the ``algebra of contextuality''.





\section{Outlook}
Non-commutativity is a fundamental mathematical feature of quantum mechanics, distinguishing it from classical physics.
But in many key cases, we do not simply have the failure of commutativity, but rather that commutativity holds \emph{up to} a specified scalar.
This is the phenomenon of \emph{commutation relations}, which
play a central role in quantum physics. There are many familiar examples.

In this paper, we have given an answer, in the discrete case working over $\ZZ_d$, to the question: what \emph{is} a commutation relation in general?
This opens up the possibility of classifying the possible contextual behaviours arising from commutation relations. By virtue of the existence of unitary representations, these arise within quantum mechanics.

We mention a few topics of current and future work:
\begin{itemize}
\item Studying the cohomology of commutation groups, and relating this to the cohomological criteria for contextuality studied e.g.~in \cite{DBLP:journals/corr/abs-1111-3620,DBLP:conf/csl/AbramskyBKLM15,okay2017topological,aasnaess2022comparing}.
\item Studying commutation groups in relation to state-dependent contextuality and empirical models \cite{abramsky2011sheaf}.
\item Relating commutation groups to the logical analysis of contextuality in terms of partial Boolean algebras \cite{kochen1975problem,DBLP:conf/csl/AbramskyB21}.
\item Generalizing commutation groups to more general abelian groups of scalars.
\item A Stone-von Neumann type theorem for commutation groups.
\item Commutation groups are highly tractable, and can be represented as subgroups of generalized Pauli groups. By contrast, solution groups \cite{cleve2017perfect} are highly expressive, and computationally intractable. Are there some natural classes of structures which are more expressive than commutation groups, while remaining computationally tractable?
\end{itemize}

\bibliography{bibfile}

\appendix

\section{The Pauli group on qubits}
We recall the definition of the \textbf{Pauli operators}, dichotomic (\ie two-valued) observables corresponding to measuring spin in the $x$, $y$, and $z$ axes, with eigenvalues $\pm 1$
\[
X\coloneqq \left(
\begin{matrix}
0 & 1 \\
1 & 0
\end{matrix}
\right)
~~
Y\coloneqq \left(
\begin{matrix}
0 & -i \\
i & 0
\end{matrix}
\right)
~~
Z\coloneqq \left(
\begin{matrix}
1 & 0 \\
0 & -1
\end{matrix}
\right)
\]
These matrices are self-adjoint, have eigenvalues $\pm 1$, and together with the identity matrix $I$ satisfy the following relations:
\begin{gather}
X^2=Y^2=Z^2=I \nonumber \\ 
XY=iZ,~~ YZ=iX,~~ ZX=iY, \label{equ: Pauli's}\\
YX=-iZ,~~ZY=-iX,~~XZ=-iY. \nonumber
\end{gather}

\end{document}